\documentclass[a4paper,12pt]{article}
\usepackage{amsmath,amssymb,amsbsy,amsfonts,amsthm,latexsym,amsopn,amstext,
            amsxtra,euscript,amscd}
            \usepackage{centernot}
\usepackage{color}
\title{{\large\bf A DSA-like digital signature protocol  \\
}}
\author{
Leila Zahhafi\footnote{leila.zahhafi@gmail.com} \ and Omar Khadir\footnote{khadir@hotmail.com}\\
Laboratory of Mathematics, Cryptography, Mechanics \\ and Numerical Analysis, Fstm\\
  University Hassan II of Casablanca,   Morocco\\
}

\date{}
\newtheorem{theorem}{Theorem}

\theoremstyle{definition}

\begin{document}
\maketitle

\baselineskip=20pt

\noindent\hrulefill

 {\abstract \small In this paper we propose a new digital signature protocol inspired
by the DSA algorithm. The security and the complexity are
analyzed. Our method constitutes  an alternative if the classical
scheme DSA is broken.

}

\noindent\hrulefill

 \vspace{0.3cm} \noindent {\small \bf Keywords :} \
{\small Public key cryptography.  Digital signature . DSA
protocol. Discrete logarithm problem. }

\vspace{0.5cm}

\noindent {{\small \bf MSC 2010 : } {\small 94A60, 11T71}

\section{Introduction}
\label{IntroductionS1}

Moderne data protection started with the work of
Shannon\cite{Shannon1949} in $1949$ on information theory. However
modern public key cryptography appeared clearly when, in $1976$
Diffie and Hellman\cite{DH1976} showed how any two network users
can construct a common secret key even if they never meet. One
year and few months later the RSA\cite{RSA1978} method was
published. It is considered as the most used cryptosystem in the
daily life.

%------------------------------------------------------\\
Digital signature is an important tool in cryptography. Its role
in funds transfert, online business, electronic emails, users
identification or documents integrity is essential. Let us recall
its mechanism principle. A trusted authority prepares the keys for
Alice who is a network user. There is a secret key $k$ and a
public key $K$ depending on her identity parameters. If she wants
to sign a document $D$, she must solve a hard problem $Pb(K,D)$
which is function of $K$ and $D$. She is able to find a solution
as she possesses a supplementary information: her private key $k$.
For anybody else, the problem, consisting generally of a difficult
equation of a high mathematical level and based on the elements
$K$ and $D$, is intractable even with the help of computers. No
one can forge Alice personal digital signature on the document
$D$. On the other hand, to validate and accept the signature,
anyone and some time it is the judge, can verify if the answer
furnished by Alice is correct or not.

%-------------------------------------------------------------
One of the first concrete digital signature protocol was proposed
by ElGamal\cite{ElGamal1985} in $1985$. It is based on the
discrete logarithm problem computationally considered as
intractable\cite[p.~103]{HAC1996}\cite[p.~236]{Stinson2006}\cite{Addepalli2017}.
Provided that the signature system parameters are properly
selected, the algorithm security of the scheme has never been
threatened.

%---------------------------------
Many variants of this scheme have been created. In $1991$, Schnorr
\cite{Schnorr1991} proposed a similar signature protocol inspired
by ElGamel system. The digital signature algorithm or
DSA\cite{DSA1991}, also deduced from the ElGamal signature
equation, was  definitively formulated by the National Institute
of Standards and Technology (NIST) in $1994$\cite{DSA1994}. Some
variants\cite{Chen2008} of the DSA were published and several
attacks were elaborated against it. In $1997$ Bellare et
al.\cite{Bellare1997} presented an attack where they showed that
if the DSA algorithm uses linear congruential pseudorandom number
generator then it is possible to recover the signer secret key. In
$2002$ Nguyen and Shparlinski\cite{Nguyen2002} published a
polynomial time algorithm that can totally break the DSA protocol
if they get few bits from the nonces used to sign a certain number
of documents. More recently Poulakis\cite{Poulakis2016} used
lattices theory to construct a system of linear equation that
leads to the disclosure of the signer secret key. At last, in
$2017$, Angel et al.\cite{Angel2017} with extensive experiments,
elaborated a method where they exploit blocks in the ephemeral
nonce and find the signer private key.

%---------------------------------------------------------------------
In this paper we propose a new digital signature scheme inspired
by the DSA algorithm. The security and the complexity are
analyzed. Our method constitues an alternative if the  classical
protocol DSA is broken. The drawback of our algorithm is that the
generation of the signature has three parameters instead of two
for DSA  and we have to execute one more modular exponentiation in
the verification step. For a theoretical and a pedagogical
interest, we present an extension of the method.

%---------------------------------------------------------------------
 The paper is organized as follows: In the next section we briefly
recall the description of the classical DSA digital signature. In
section $3$ we present our contribution and we conclude in section
$4$.

%-----------------------------------------------------------------
Classical notations will be adopted. So $\mathbb N$ is the set of
all natural integers. When $a,b,n\in\mathbb N$, we write $a=b\
mod\ n$ if $a$ is the remainder of the division of the integer $b$
by $n$, and $a\equiv b\ [n]$ if the number $n$ divides the
difference $a-b$. If $p$ is a prime integer then the set
$(\dfrac{\mathbb Z}{p\mathbb Z})^*$ is the multiplicative group of
modular integers $\{1,2,\ldots,p-1\}$.

\noindent We begin by recalling the DSA signature method.

\section {The standard DSA protocol\cite{DSA1994,DSA1991} } In this section, we describe the basic
DSA scheme followed by the security analysis of
the method.
\subsection{Keys production}
The signer  selects  two primes $p$ and $q$ such that $q$ divides
$p-1$, $ 2^{t-1} < q < 2^{t}$ with:
$t\in \{160,256,384,512\}$,  $ 2^{L-1} < p < 2^{L}$, $ 768 < L < 1024$ and $L$ is a multiple of 64.\\
Then, he chooses a primitive root  $g$ mod $p$ and computes
$\alpha = g^{ \frac{p-1}{q}} \: mod \:  p$. The signer selects
also an integer $x$ such that $1 \leq x \leq q-1$ and calculates $
y = \alpha^{x} \: mod \: p $. Finally, he publishes $(p,q,g,y)$
and keeps the parameter $x$ secret as its private key.

\subsection{The signature generation}
Let $h$ be a secure hash function\cite[p.~33]{HAC1996} that produces a 160-bit output.\\
 To sign a message
$m$, the signer starts by selecting a random secret integer $k$
smaller than $q$ and called the nonce\cite[p.~397]{HAC1996}. Then,
he computes successively $r = (\alpha ^{k} \: mod \: p) \: mod \:
q $ and  $s = \displaystyle \frac{h(m) + a.r}{k} \ mod \: q $.
Finally, the signature is the pair $(r,s)$.

\subsection{The signature verification}
The verifier of the signature calculates: $u_{1} = \displaystyle
\frac{h(m)}{s} \: mod \: q$ and  $u_{2} = \displaystyle
\frac{r}{s} \: mod \: q$.
Then, he computes: $ v = ((\alpha^{u_{1}}y^{u_{2}}) \: mod \: p)\: mod \: q $.\\
Depending on $v = r$ or not,  he accepts or rejects the signature.

\subsection{Security of the method}
To date the DSA system is considered as a secure digital signature scheme. Indeed to break it,  attackers  must first solve a
famous hard mathematical question: the discrete logarithm problem (DLP)\cite[p.~103]{HAC1996}\cite[p.~267]{Stinson2006}. Conversely,
we ignore whether or not breaking the DSA scheme leads to an algorithm for solving the DLP. It's a remarkable open problem.
The introduction of the second large prime $q$ in the DSA mechanism  avoids Pohlig and Hellman attack\cite{Pohlig1978}
on the discrete logarithm problem.\\
Several attacks were mounted against DSA protocol. The reader
is invited to see for instance references\cite{Angel2017,Poulakis2016, Nguyen2002,Bellare1997}.
 It's well known\cite[p.~188]{Mollin2007} that, as for ElGamal signature scheme\cite{ElGamal1985}, using the same nonce $k$
  to sign two different documents  reveals the system secret key. Therefore it's mandatory to change the value of $k$ at each signature.

%---------------------------------------------------- OUR SIGNATURE PROTOCOLE----------------------
\noindent The following section describes our main contribution.

\section{A DSA-like digital signature protocol}

In this section, we present a new  DSA-like digital signature and
we analyze its security and complexity.

\subsection{Keys production}
The fabrication of the keys is the same as for the DSA protocol.
The signer selects two primes $p$ and $q$ such that $q$ divides
$p-1$, $ 2^{t-1} < q < 2^{t}$ with:
$t\in \{160,256,384,512\}$,  $ 2^{L-1} < p < 2^{L}$, $ 768 < L < 1024$ and $L$ is a multiple of 64.\\
Then, he chooses a primitive root  $g$ mod $p$ and computes
$\alpha = g^{ \frac{p-1}{q}} \: mod \:  p$. The signer selects
also an integer $x$ such that $1 \leq x \leq q-1$ and calculates $
y = \alpha^{x} \: mod \: p $. Finally, he publishes $(p,q,g,y)$
and keeps the parameter $x$ secret as its private key.

\subsection{The signature generation}

Let $h$ be a collision resistant hush
function\cite[p.~323]{HAC1996} such that the image of any message
is belonging to the set $\{1,2,\ldots,q\}$. \\
To sign the  document  $m$, Alice must solve the modular signature
equation:\begin{equation}\label{eq1}
\alpha^{\frac{h(m)}{t}}y^{\frac{r}{t}}r^{\frac{s}{t}}\  mod\  p
\equiv s\ [q]
\end{equation}
where the unknown parameters $r,s,t$ verify $0<r<p$ et $0<s,t<q$.
\begin{theorem}
Alice, with her secrete key $x$, is able to find a triplet (r,s,t)
that verifies the signature equation \eqref{eq1}.
\end{theorem}
\begin{proof}
Alice chooses two random numbers $k,l<q$ then computes
$r=\alpha^k\ mod\ p$ and $s=\alpha^l\ mod\ p\ mod\ q$. To get a
solution of equation \eqref{eq1} it suffices to have
$\alpha^{\frac{h(m)}{t}}\alpha^{x\frac{r}{t}}\alpha^{k\frac{s}{t}}\equiv
\alpha^l\  [p]$. On the other hand the third parameter $t$ must
verify $\dfrac{h(m)+xr+ks}{t}=l\ [q]$ or
\begin{equation}\label{eq2}
t=\dfrac{h(m)+xr+ks}{l}\ mod\ q
\end{equation}

\end{proof}

\subsection{The signature verification}

To verify Alice signature $(r,s,t)$, Bob should do the
following:\\
1. He first finds Alice public key $(p,q,\alpha,y)$.\\
2. He verifies that $0<r<p$ and $0<s,t<q$. If not he rejects the
signature.\\
3. He computes $u_1 = \displaystyle \frac{h(m)}{t} \: mod \: q$,
$u_2 = \displaystyle \frac{r \: mod \: q}{t} \: $ $ mod \: q$ and
$u_3 = \displaystyle \frac{s
}{t} \: mod \: q$. \\
4. He determines $v = ((\alpha^{u_1}y^{u_2}r^{u_3}) \: mod \: p)\:
mod \: q$
\\
5. If $v = s$, Bob accepts the signature, otherwise  he rejects
it.

 \vspace{0.3cm} \noindent Before going on, we illustrate the procedure by an example. We took
the same two large primes  $p$ and $q$ and the generator $g$ from
reference\cite{Pornin2013}.

\vspace{0.3cm}\noindent{\bf Example.}
Number $p$ is the $1024$ bit-length prime:\\
$p=94772214835463005053734612688987420745400704367641322402256120323201\\
101201888870137170705373498571313030316679878174804574980244779195907606\\
098768964031739134779278848798229819934901324222106210711842549374102491\\
417296346772453897799554117544442700769168461664359227744193913924495898\\
621041399925210910234489$, \\
and $q=875964080856129786106302881659054003458244253873$.\\
The generator of the multiplicative group $((\dfrac{\mathbb Z}{p\mathbb Z})^*,.)$ is\\
$g=5401015700248054412670727427618885968561082170092747201289836023587070\\
84212996041390972417923715770599593251420677788894704063133555201170818988\\
61805293399530415510300321792173737741806916560867092300436919535062464506\\
65697269251181380801651442389601287318661304902597519067842079816229492516\\
762912476306877$.\\
 We find that:\\
$\alpha=527223567708677258197455859133222937206317701786932418731523516024196571\\
7085174644862184855390335598300773471653030132135612590607930128253659134294\\
8495010997120001192272226027997156250022237643201186825665260281160455408116\\
6154971036251524872300380557802833133742393048305876853382003116504889806209\\
1119992$.\\
Suppose that the secrete key is:\\
$x=371575259833906365510684947508061994685469500919$, so:\\
$y=630775473718828368762466985220658584945893263294126314383120993968845387\\
9656101046604281277749079342479732012406513115015034410271324671714864739538\\
5456320857555381477923626248021182202362864702913038804579836525851535288840\\
6082786247355168380930043030984132308743260982091689595031807130951481230017\\
18879430$.\\
Assume that Alice decides to sign the message $m$ such that
$h(m)=123456789$.
She uses the random secrete exponents $k=1250$ and  $l=98561$. Therefore:\\
$r=578164865504179483400175892277942835618633089799304565306500743940987577\\
3739310417729452466781593533114855965860037328909457721341899371409595718096\\
8602212875285420808831668566865480631794760348104088294970278483930299680828\\
49153157235106512542837849342137689407189514011990142602691770559506\\
0711695234702599$; \\
$s=544621099954698016824748794802717914623249907270$ and\\
$t=556119013460694353294511753174948468444082504155$.\\
To test the validity of Alice signature Bob first determines $u_1$, $u_2$ and $u_3$: \\
$u_1=141694501602616348876138369891969021089672807186$;\\
$u_2=662220963645670062957535725628437873682297068731$;\\
$u_3=68770021753905823557652121472773639225383494491$; \\
then Bob calculates $v=\alpha^{u_1}y^{u_2}r^{u_3}\  mod\  p\  mod\
q$\\
$=544621099954698016824748794802717914623249907270$ which is
exactly $s$. In this case,  Alice  signature is  accepted.

\subsection{Security analysis}
\noindent We discuss here some possible attacks. Suppose that Oscar,
enemy of Alice, tries to impersonate her by signing the message $m$ without knowing her secret key $x$.\\
{\bf Attack 1.} After receiving the signature parameters $(r,s,t)$
of a particular message $m$, Oscar may be wants to find Alice
secret key $x$.  If he uses equation\eqref{eq1}, he will be
confronted to the discrete logarithm problem $a^x\equiv b\ [q]$
where $a=\alpha^{r/t}\ mod\ p$ and $b=s\alpha^{-h(m)/t}\,r^{-s/t}\
mod p$. If Oscar prefers to exploit relation \eqref{eq2}, he needs
to know the two nonces $k$ and $l$. Their computation derives from the discrete logarithm problem.\\
{\bf Attack 2.} Assume now that Oscar arbitrary fixes random
values for two parameters and tries to find the third one. \\
(i) If he fixes $r$ and $s$ in the signature equation \eqref{eq1},
and likes to determine the third unknown parameter $t$,
he will be confronted to the discrete logarithm problem  $a^{t'}\equiv s\ [p]$ where $a=\alpha^{h(m)}y^rr^s\ mod\ p$ and $t'=\dfrac{1}{t}\ mod\ q$.\\
(ii) If Oscar fixes $r$ and $t$ and wants to find the parameter
$s$, he will be confronted to the equation $ab^s\equiv s\ [q]$,
where $a=\alpha^{\frac{h(m)}{t}}y^{\frac{r}{t}}\ mod\ p$ and
$b=r^{1/t}\ mod\ p$. In the mathematical literature, we don't know
any algorithm to solve this kind of modular equation.\\
(3) If Oscar fixes $s$ and $t$ then likes to calculate the
parameter $r$, he must solve the equation $a^rr^b\equiv c\ [q]$,
where $a=y^{1/t} \ mod\ p$, $b=\dfrac{s}{t}\ mod\  q$ and
$c=s\alpha^{-{h(m)}/t}\ mod\ p$.
 There is no known general method for solving this problem.\\
 {\bf Attack 3.} Assume that Alice used the same couple of exponents $(k,l)$ to sign two distinct message $m_1$ and $m_2$. Being aware of this fact, Oscar,
 from the first message signature obtains $lt_1\equiv h(m_1)+xr_1+ks_1\ [q]$ and from the second message $lt_2\equiv h(m_2)+xr_2+ks_2\ [q]$. As $r_1=r_2$ and
 $s_1=s_2$, Oscar is able to calculate the nonce $l$. In contrast to the ElGamal and DSA schemes, it
 seems that there is no easy way to compute the exponent $k$ and then to retrieve Alice secret key
 $x$.\\
 {\bf Attack 4.} Let $n\in\mathbb N$. Suppose that Oscar has collected $n$ valid signatures $(r_i,s_i,t_i)$ for
messages $m_i$, $i\in\{1,2,\ldots, n\}$. Using \eqref{eq2}, he
will  construct a system of $n$ modular equations:\\
 $$
(S) \left\{%
\begin{array}{c}
  l_1t_1\equiv h(m_1)+xr_1+k_1s_1\ [q] \\
  l_2t_2\equiv h(m_2)+xr_2+k_2s_2\ [q] \\
  \vdots\ \ \ \ \ \ \vdots\ \ \ \ \ \  \vdots  \\
  l_nt_n\equiv h(m_n)+xr_n+k_ns_n\ [q] \\
\end{array}%
\right.$$ where $\forall i\in\{1,2,\ldots, n\}$,
$r_i=\alpha^{k_i}\ mod\ p$ and $s_i=\alpha^{l_i}\ mod\ p\ mod\
q$.\\
Since system (S) contains $2n+1$ unknown parameters  $x, r_i, s_i,
i\in\{1,2,\ldots, n\}$, it is not difficult for Oscar to propose a
valid solution. But Alice secret key $x$ has a unique possibility
and therefore Oscar will never be sure what value of x
is the right one. So this attack is not efficient.\\
 {\bf Attack 5.} Let us analyze the existential forgery\cite[p.~285]{Stinson2006}.
 Suppose that the signature protocol is used without the hash function $h$.
 Oscar can put $r=\alpha^k\,y^{k'}\ mod\ p$, $s=\alpha^l\,y^{l'}\ mod\
 p\ mod\ q$ for arbitrary numbers $k,k',l,l'$. To solve the signature equation \eqref{eq1}, it suffices to solve the
 system:\\
$\displaystyle \left\{%
\begin{array}{c}
  \dfrac{m}{t}+k\dfrac{s}{t}\equiv l\ [q] \\
  \dfrac{r}{t}+k'\dfrac{s}{t}\equiv l'\ [q] \\
\end{array}%
\right.$, so $\displaystyle \left\{%
\begin{array}{c}
 m\equiv tl-ks\ mod\ q \\
  t=\dfrac{1}{l'}[r+k's]\ mod\ q \\
\end{array}%
\right.$. Hence $(r,s,t)$ is a valid signature for the message
$m$, but this attack is not realistic.\\
 {\bf Attack 6.}  Suppose that Alice enemy Oscar is able to break
 the DSA scheme. In another words, given $p,q,m, \alpha, y$, he
 can find integers $r,s<q$ such that $\alpha^{\frac{h(m)}{s}}\,y^{\frac{r}{s}}\
 mod \ p\equiv s\ [q]$. There is no evidence that Oscar is able to
 solve equation \eqref{eq1} $\alpha^{\frac{h(m)}{t}}y^{\frac{r}{t}}r^{\frac{s}{t}}\  mod\  p
\equiv s\ [q]$. It's an advantage of our signature model: breaking
the DSA scheme does not lead to breaking our protocol.

\noindent {\bf Remark 1.} We end this security analysis by asking
a question for which we have no answer: If someone is able to
simultaneously break ElGamal, DSA and our own signature protocols,
can he solve the general discrete logarithm problem ?
\subsection{Complexity}
Productions of  public and private keys in our protocol and in the
DSA scheme are identical. So the number of operations to be
executed is the same. In the generation of the signature
parameters, we have one more parameter than in the DSA. To compute
it, we use a a supplementary modular exponentiation. For the
verification step, we calculate three exponentiation
instead of two for the DSA scheme.\\
Let $T_{exp}$ and $T_{mult}$ the times necessary to compute
respectively an exponentiation and a multiplication. The total
time to execute all operations using our method is as follows:
\begin{equation}\label{eq3}
T_{tot} = 7T_{exp} + 8T_{mult}
\end{equation}
As $T_{exp} = O(\log^3 n)$ and $T_{mult} = O(\log^2 n)$, (see
\cite[p.~72]{HAC1996}),  the final complexity of our signature
scheme is
\begin{equation}\label{eq4}
T_{tot} =  O(\log^2 n + \log^3 n)=O(\log^3 n)
\end{equation}
 This proves that the execution of the protocol works in a polynomial time.

\subsection{Theoretical generalization}
For it's pedagogical and mathematical interest, we end this paper
by giving an extension of the signature equation \eqref{eq1}.
Let $h$ be a known and secure hash function as mentioned in section $2$ and in the beginning of this section.\\
We fix an integer $n\in\mathbb N$ such that $n\geq 2$.\\
{\bf 1.} Alice begins by choosing her public key $(p, q,\alpha,y)$, where $p$ and $q$ primes such that $q$ divides $p-1$.\\
Element $\alpha$ is a generator of the subgroup of $\displaystyle
\frac{\mathbb Z}{p\mathbb Z}$ whose order  is $q$. $y=\alpha^x\
mod\ p$ where $x$ is a secret parameter
in $\{1,2,\ldots,q-1\}$. Integer $x$ is  Alice private key.\\
{\bf 2.} If Alice likes to product a digital signature of a message $m$, she must solve the congruence:
\begin{equation}\label{eq5}
\alpha^{\frac{h(m)}{r_{n+1}}}\,y^{\frac{r_1}{r_{n+1}}}\,
{r_1}^{\frac{{r_2}}{r_{n+1}}}\,
{r_2}^{\frac{{r_3}}{r_{n+1}}}\,\ldots
{r_{n-1}}^{\frac{{r_n}}{r_{n+1}}}\ mod\ p \equiv r_n\ [q]
\end{equation}
where the unknown parameters $r_1,r_2,\ldots, r_{n+1}$ verify
\begin{equation}\label{eq6}
 0<r_1,r_2, \ldots, r_{n-1}<p {\rm \ and\ }
0<r_n,r_{n+1}<q.
\end{equation}
\begin{theorem}
Alice, with her secrete key $x$ can determine an $(n+1)$uplet
$(r_1,r_2,\ldots,r_n,r_{n+1})$ that verifies the modular relation
\eqref{eq5}.
\end{theorem}
\begin{proof}
The signer Alice selects $n-1$ random numbers
$k_1,k_2,\ldots,k_{n-1}\in \mathbb N$ lm2ess than the prime $q$
then computes $r_i=\alpha^{k_i}\ mod\ p$ for every
$i\in\{1,2,\ldots,n-1\}$ and $r_n=\alpha^{k_n}\mod\ p\ mod\ q$. We
have:\\
Equation \eqref{eq5} $ \Longleftrightarrow
\alpha^{\frac{h(m)}{r_{n+1}}}\, \alpha^{x\frac{r_1}{r_{n+1}}}\,
\alpha^{{\displaystyle  \sum_{i=1}^{n-1}}\,
k_i\frac{r_{i+1}}{r_{n+1}}}\
mod\ p\ mod\ q=\alpha^{k_n}\ mod\ p\ mod\ q$.\\
It suffices to have $\alpha^{\frac{h(m)}{r_{n+1}}}\,
\alpha^{x\frac{r_1}{r_{n+1}}}\, \alpha^{{\displaystyle  \sum_{i=1}^{n-1}}\,
k_i\frac{r_{i+1}}{r_{n+1}}}\ mod\
p=\alpha^{k_n}\ mod\ p$, which is equivalent to :\\
$\dfrac{h(m)}{r_{n+1}}+x\frac{r_1}{r_{n+1}}+\displaystyle
\sum_{i=1}^{n-1}\, k_i\frac{r_{i+1}}{r_{n+1}}\equiv k_n\ [q]$. So Alice determines the last unknown parameter $r_{n+1}$ by calculating
\begin{equation}\label{eq6}
r_{n+1}=\frac{1}{k_n}[h(m)+xr_1+\sum_{i=1}^{n-1}\, k_ir_{i+1}]\ mod\ q
\end{equation}
\end{proof}
{\bf 3.} If Bob receives from Alice her signature preuve $(r_1,
r_2,\ldots,r_n,r_{n+1})$, he will be able to check whether the
modular equation \eqref{eq5} is valid or not. He then deduces if
he accepts or rejects this signature.

\noindent {\bf Remark 2.} Let $k_0$ be Alice secret key $x$,
${\overrightarrow{u}}$ and ${\overrightarrow{v}}$  respectively
the vectors $(k_0,k_1,\ldots,k_{n-1})$ and $(r_1,r_2,\ldots,r_n)$.
To easily memorize equality \eqref{eq6}, observe that
$r_{n+1}=\dfrac{h(m)+\overrightarrow{u}.\overrightarrow{v}}{k_n}$
where $\overrightarrow{u}.\overrightarrow{v}$ denotes the
classical inner product of the two vectors $\overrightarrow{u}$
and $\overrightarrow{v}$.

\section{Conclusion}
\label{ConclusionS4} In this article, a new digital signature
protocol was presented. We studied in details the security of the
method and gave an analysis of its complexity. Our contribution
can be seen as an alternative if the DSA algorithm is totally
broken. For its purely mathematical and pedagogical interest, we
furnished a general form of our proposed signature equation.

\end{document}